\documentclass[]{statsoc}

\usepackage{color,latexsym,amsfonts}
\usepackage{amssymb}
\usepackage{amsmath}
\usepackage{bbm}
\usepackage{graphicx}
\usepackage{float}

\graphicspath{{Figures/}}

\def\blue#1{\textcolor{blue}{#1}}

\def\mH{\mathcal{H}}
\def\mI{\mathcal{I}}
\def\sX{\mathcal{X}}

\def\Re{\mathbb{R}}

\def\bA{\boldsymbol{A}}
\def\balp{\boldsymbol{\alpha}}
\def\halp{\hat{\alpha}}
\def\bB{\boldsymbol{B}}
\def\hbeta{\hat{\beta}}
\def\bc{\boldsymbol{c}}
\def\bC{\boldsymbol{C}}
\def\bd{\boldsymbol{d}}
\def\bD{\boldsymbol{D}}
\def\fs{f^{\star}}
\def\bgam{\boldsymbol{\gamma}}
\def\bLam{\boldsymbol{\Lambda}}
\newcommand{\bone}{\boldsymbol{1}}
\def\bP{\boldsymbol{P}}
\def\bpsi{\boldsymbol{\psi}}
\def\bQ{\boldsymbol{Q}}
\def\br{\boldsymbol{r}}
\def\bR{\boldsymbol{R}}
\def\bS{\boldsymbol{S}}
\def\bth{\boldsymbol{\theta}}

\newcommand{\trans}{^{\mbox{\tiny {\sf T}}}}

\def\bz{\boldsymbol{z}}
\def\bZ{\boldsymbol{Z}}

\DeclareMathOperator*{\argmin}{arg\,min}

\newtheorem{lemma}{Lemma}
\newtheorem{corollary}{Corollary}

\newtheorem{theorem}{Theorem}

\def\E{{\rm E}\,}

\usepackage{mathtools}

\def\real{\mathbb{R}}
\def\natural{\mathbb{N}}
\def\lmb{\lambda}
\DeclarePairedDelimiter\norm{\lVert}{\rVert}
\DeclarePairedDelimiter\abs{\lvert}{\rvert}
\def\iprod#1#2{\left\langle#1,#2\right\rangle}
\def\holders{H{\"o}lder's }
\def\xv{\mathrm{E}}
\def\prob#1{\mathrm{P}\left(#1\right)}
\def\Indc#1{\boldsymbol{1}\!\left[#1\right]}
\def\xval#1{\mathrm{E}\left(#1\right)}
\def\veps{\varepsilon}

\title[RKHS for Functional Classification]{A reproducing kernel Hilbert space framework for functional classification}
\author[Sang, Kashlak, Kong]{Peijun Sang}
\address{University of Waterloo,
	Waterloo, ON,
	Canada}
\email{peijun.sang@uwaterloo.ca }
\author[Sang, Kashlak, Kong]{Adam B Kashlak, Linglong Kong}
\address{University of Alberta,
	Edmonton, AB,
	Canada}

\begin{document}

\begin{abstract}{
 We encounter a bottleneck when we try to borrow the strength of classical classifiers to classify functional data. The major issue is that functional data are intrinsically infinite dimensional, thus classical classifiers cannot be applied directly or have poor performance due to curse of dimensionality. To address this concern, we propose to project functional data onto one specific direction, and then a distance-weighted discrimination DWD classifier is built upon the projection score. The projection direction is identified through minimizing an empirical risk function that contains the particular loss function in a DWD classifier,
 over a reproducing kernel Hilbert space. Hence our proposed classifier can avoid overfitting and enjoy appealing properties of DWD classifiers. This framework is further extended to accommodate functional data classification problems where scalar covariates are involved. In contrast to previous work, we establish a non-asymptotic estimation error bound on the relative misclassification rate. In finite sample case,  we demonstrate that the proposed classifiers compare favorably with some commonly used functional classifiers in terms of prediction accuracy through simulation studies and a real-world application.
}

\end{abstract}

\section{Introduction}
For functional data classification, the explanatory variable is usually a random function and the outcome is a categorical random variable which can take two or more categories. As with classification problems for scalar covariates, a functional classifier is built upon a collection of observations consisting of a functional covariate and a categorical response for each subject, and then a class label will be assigned to a new functional covariate based on the classifier.  A typical example is the phoneme recognition problem in \cite{friedman2001elements}. Log-periodograms for each of the two phonemes ``aa" and ``ao" were measured at 256 frequency levels, and the primary goal is to make use of these log-periodograms to classify phonemes.  For this problem, the log-periodograms measured at 256 frequencies can be regarded as a functional covariate and the outcome takes two possible categories: ``aa" or ``ao". Therefore, this phoneme recognition can be framed as a functional classification problem. Actually, functional classification has been extensively studied in the literature due to its wide applications in various fields such as neural science, genetics, agriculture and chemometrics (\citealp{tian2010functional, leng2005classification, delaigle2012achieving, berrendero2016variable}).

As pointed out by \cite{fan2008high}, a high dimension of scalar covariates would yield a negative impact on prediction accuracy of classifiers due to curse of dimensionality. This issue is even more serious in functional classification since functional data are intrinsically infinite dimensional \citep{ferraty2004}. In light of this fact, dimension reduction was suggested before classifying functional data. Functional principal component (FPC) analysis a commonly used technique in this regard, and various classifiers for functional data have been proposed based on FPC scores, which are the projections of functional covariates onto a number of FPCs. Typical examples include discriminant analysis \citep{hall2001functional},
naive Bayes classifier  \citep{dai2017optimal} and logistic regression \citep{leng2005classification}. Since FPC analysis is an unsupervised dimension reduction approach, the retained FPC scores are not necessarily more predictive of the outcome compared with the discarded ones. In contrast, 
treating fully observed functional data as a random variable in a Hilbert space without dimension reduction has also attracted substantial attention in functional classification \citep{yao2016probability}. For instance, \cite{ferraty2003curves} proposed a distance-based classifier for functional data, and \cite{biau2005functional} and \cite{cerou2006nearest} considered $k$-nearest neighbor classification.

An optimal separating hyperplane can be constructed to distinguish two perfectly separately classes. This idea is further extended to accommodate the nonseparable case in support vector machines (SVM). More specifically, with the aid of the so-called kernel trick, the original feature space is expanded and a linear boundary in this expanded feature space can separate the two overlapping classes very well. If we project this linear boundary back onto the original feature space, it is a nonlinear decision boundary. For a more comprehensive introduction to the SVM, one can refer to \cite{vapnik2013nature} and \cite{cristianini2000introduction}. Due to the ability of constructing a flexible decision boundary, different versions of SVM for functional data have been proposed in the literature. \cite{rossi2006support} considered projecting functional covariates onto a set of fixed basis functions first, and then applied SVMs to the projections for classification. Nevertheless, \cite{yao2016probability} proposed a supervised method to perform dimension reduction for functional data. Weighted SVM \citep{lin2002support} were then constructed on the reduced feature space.
 \cite{wu2013functional} recovered trajectories of sparse functional data or longitudinal data using the principal analysis by conditional expectation (PACE) \citep{yao2005functional} first, and then proposed a support vector classifier for the random curves. But the convergence rate of the SVMs with a functional covariate was not established in the aforementioned work.

 \cite{marron2007distance} noted that the data piling problem may cause a deterioration in performance of SVM. They proposed the distance-weighted discrimination (DWD) classifier that can make uses of all observations in a training sample, rather than the support vectors in SVM, to determine the decision boundary. \cite{wang2018another} proposed an efficient algorithm to solve the DWD problem. In this article, we extend the idea of the DWD to functional data to address a binary classification problem. The basic idea is to find an optimal projection direction such that the DWD classifier built upon this projected score achieves good prediction performance. Additionally, to avoid overfitting on the training sample, we incorporate a roughness penalty term when minimizing the empirical risk function. Penalized approaches have been investigated recently in the context of functional linear regression. Interested readers can refer to the work by \cite{yuan2010reproducing} and \cite{sun2018optimal}.  However, as far as we know, this framework has received little attention in functional classification problems. The method for functional classification proposed in this article estimates the slope function through seeking a minimizer of a regularized empirical risk function over a reproducing kernel Hilbert space (RKHS). This RKHS is closely associated with the penalty term in the regularized empirical risk function. With the help of the representer theorem, we are able to convert this infinite dimensional minimization problem to a finite dimensional problem. This fact lays the foundation of numerical implementations for the proposed classifier. This framework is further extended to accommodate classifications when observations of both a functional covariate and several scalar covariates are available for each subject. There has been extensive research on partial functional linear regression models for such scenarios; see \cite{kong2016} and \cite{wong2019partially} for instance. However, we have not seen much progress made for classifications in this regard; thus our work fills the gap of functional data classification when scalar covariates are also available. 
In addition to the novel methodology, we establish a non-asymptotic ``oracle type inequality" for the bound on the convergence rate of the relative loss and the relative classification error. This error bound is essentially different from those considered in \cite{delaigle2012achieving}, \cite{dai2017optimal} and \cite{berrendero2018use}, all of which focused on asymptotic perfect classification.

The rest of this article is organized as follows. In Section \ref{sec-method} we introduce the RKHS-based functional DWD classifier for classifying functional data without and with scalar covariates. Theoretical properties of the proposed classifiers are established in Section \ref{sec-theory}. We carry out simulation studies in Section \ref{sec-simul} to investigate the finite sample performance of the proposed classifiers in terms of prediction accuracy. In Section \ref{sec-real} we consider one real world application to demonstrate the performance of the proposed classifiers. We conclude this article in Section \ref{sec-end}. 
All technical proofs are provided in the Appendix.

\section{Methodology} \label{sec-method}
Let $X$ denote a random function  with a compact domain $\mI$, and $Y$ is a binary outcome related to $X$.
Without loss of generality, we assume that $Y \in \{-1, 1\}$. Suppose that the training sample consists of $(x_i, y_i), i = 1, \ldots, n$, $n$ i.i.d. copies of $(X, Y)$. Our primary goal is to build a classifier based on this training sample.

We first present an overview of the distance-weighted discrimination (DWD) proposed by \cite{marron2007distance}. Consider the following classification problem where $\bz = (z_1, \ldots, z_p)^T \in \mathcal{Z}$ be a vector of $p$ scalar covariates and $y \in \{-1, 1\}$ is a binary response. The main task is to build a classifier: $f: \mathcal{Z} \rightarrow \{-1, 1\}$ based on $n$ pairs of observations $(\bz_i, y_i), i = 1, \ldots, n$. According to \cite{wang2018another}, the decision boundary of a generalized distance-weighted discrimination classifier
can be obtained by solving
$$
\min_{\alpha_0, \balp} n^{-1} \sum_{i = 1}^n V_q\{y_i(\alpha_0 + \bz_i^T\balp)\} + \lambda \balp^T\balp,
$$
where
\begin{equation}
V_q(u) = \begin{cases} 1 - u &\mbox{if } u \leq \frac{q}{1+q}, \\
\frac{1}{u^q}\frac{q^q}{(q+1)^{q+1}} & \mbox{if } u > \frac{q}{1+q}, \end{cases}
\label{eq-Vq}
\end{equation} 
is the loss function and $\lambda > 0$ is a tuning parameter. Note that as $q \rightarrow \infty$, the generalized DWD loss function converges to the hinge loss function. $H(u) = \max(0, 1 - u)$, used in the SVM. This relationship is also illustrated in Figure \ref{fig:loss}. Denote by $(\halp_0, \hat{\balp})$ the solution to the minimization problem above.  Given a new observation $\tilde{\bz} \in \mathcal{Z}$, the predicted class label will be 1 if $\halp_0 + \tilde{\bz}^T \hat{\balp} > 0$ and -1 otherwise.

\begin{figure}[H]
	\centering
	{\includegraphics[width=6cm]{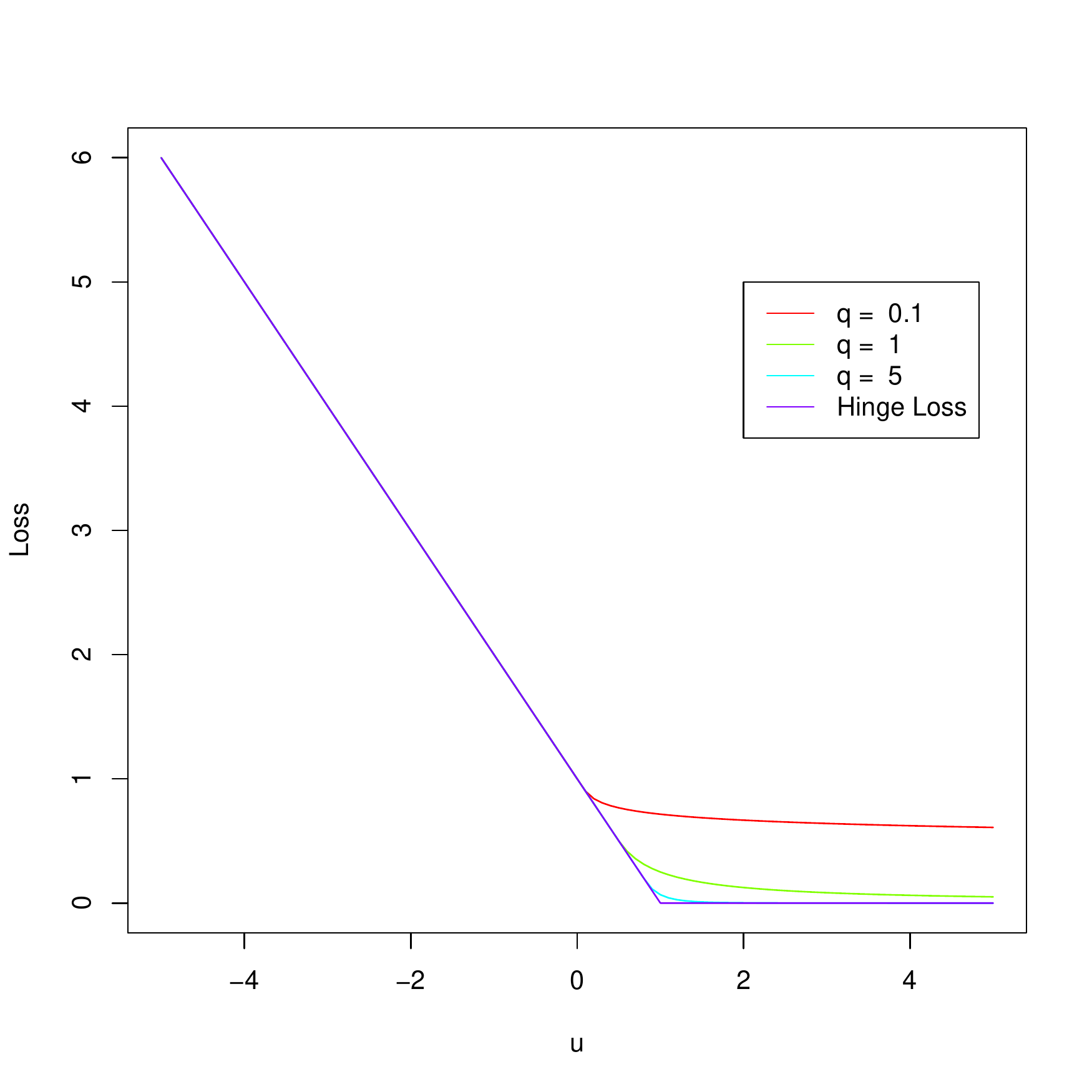}}
	\caption{The loss functions of the generalized distance-weighted discrimination with different values of $q$ and the hinge loss function.}
	\label{fig:loss}
\end{figure}

In this article, we aim to extend the framework of DWD to functional data. In particular, we consider the following objective function
\begin{equation}
Q(\alpha, \beta) = n^{-1} \sum_{i = 1}^n V_q\left\{y_i\left(\alpha + \int_{\mI} x_i(t) \beta(t) dt\right)\right\} + \lambda J(\beta),
\label{eq-target}
\end{equation}
where $J$ is a penalty functional. The penalty functional can be conveniently
defined through the slope function $\beta$ as a squared norm or semi-norm
associated with $\mH$. A canonical example of $\mH$ is the Sobolev space. Without loss of generality, assuming that $\mI = [0, 1]$, the
Sobolev space of order $m$ is then defined as
$$
\mathcal{W}_2^m([0, 1]) = \{h: [0, 1] \rightarrow \Re, h, h^{(1)},\ldots, h^{(m - 1)} \text{are absolutely continuous and}~ h^{(m)} \in L_2[0, 1] \}.
$$
Endowed with the (squared) norm
$$
||h||^2_{\mathcal{W}_2^m} = \sum_{l = 0}^{m - 1} \left(\int_0^1 h^{(l)}(t) dt\right)^2 + \int_0^1 \left(h^{(m)}(t)\right)^2dt,
$$
$\mathcal{W}_2^m([0, 1])$ is a reproducing kernel Hilbert space. In this case, a possible choice of the penalty functional is given by
\begin{equation}
J(\beta) = \int_0^1 \{\beta^{(m)}(t)\}^2 dt.
\label{eq-pen}
\end{equation}

\subsection{Representer theorem}
Let the penalty functional $J$ be a squared semi-norm on $\mH$ such that the null
space
\begin{equation}
\mH_0 = \{\beta \in \mH: J(\beta) = 0\}
\end{equation}
is a finite-dimensional linear subspace of $\mH$. Denote by $\mH_1$ its orthogonal complement in $\mH$ such that $\mH = \mH_0 \oplus \mH_1$. That is, for any $\beta \in \mH$ , there exists a unique decomposition $\beta = \beta_0 + \beta_1$ such that $\beta_0 \in \mH_0$ and $\beta_1 \in \mH_1$. Note that $\mH_1$ is also a reproducing kernel Hilbert space with the inner product of $\mH$ restricted to $\mH_1$.
Let $K: \mI \times \mI \rightarrow \Re$ be the corresponding reproducing kernel of $\mH_1$ such that $J(\beta_1) = ||\beta_1||_K^2 = ||\beta_1||_{\mH}^2$ for any $\beta_1 \in \mH_1$. Let $N = \text{dim}(\mH_0)$ and $\psi_1, \ldots, \psi_N$ be the basis functions of $\mH_0$.

We will assume that $K$ is continuous and square integrable. With slight abuse of notation, write
\begin{equation}
(Kf)(\cdot) = \int_{\mI} K(\cdot, s)f(s)ds.
\label{eq-operator}
\end{equation}
According to \cite{yuan2010reproducing}, $Kf \in \mH_1$ for any $f \in L_2(\mI)$ and $\forall \beta \in \mH_1$,
\begin{equation}
\int_{\mI} \beta(t)f(t)dt = \langle Kf, \beta \rangle_{\mH}.
\label{eq-kernel}
\end{equation}
With these observations, we are able to establish the following theorem which is crucial in numerical implementations and the theoretical analysis.

\begin{theorem}
	Let $\halp_n$ and $\hbeta_n$ be the minimizer of \eqref{eq-target} and $\hbeta_n \in \mH$. Then there exist
	$\bd = (d_1, \ldots, d_N)^T \in \Re^N$ and $\bc = (c_1, \ldots, c_n)^T \in \Re^n$ such that
	\begin{equation}
	\hbeta_n(t) = \sum_{l = 1}^N d_l\psi_l(t) + \sum_{i = 1}^n c_i(Kx_i)(t).
	\label{eq-sol}
	\end{equation}
\label{th-rep}
\end{theorem}

\subsection{Estimation algorithm}
\label{sec-est}

For the purpose of illustration, we assume that $\mH = \mathcal{W}_2^2$ and $J(\beta) = \int_0^1 (\beta^{''})^2 dt$, in the following numerical implementations. Then $\mH_0$ is the linear space spanned by $\psi_1(t) = 1$ and $\psi_2(t) = t - 0.5$. A possible choice for the reproducing kernel associated with $\mH_1$ is
$$
K(s, t) = k_2(s)k_2(t) - k_4(|s - t|),
$$
where $k_2(s) = \frac{1}{2}\left(\psi_2^2(s) - \frac{1}{12}\right)$ and $k_4(s) =  \frac{1}{24}\left(\psi_2^4(s) - \frac{\psi_2^2(s)}{2} +  \frac{7}{240}\right)$. You can refer to Chapter 2.3 of \cite{gu2013smoothing} for more details. Based on Theorem \ref{th-rep}, we only need to consider $\beta$ that takes the following form:
$$
\beta(t) = d_1 + d_2 (t - 0.5) + \sum_{i = 1}^n c_i \int_{\mI} x_i(s) K(t, s)ds
$$
for some $\bd = (d_1, d_2) \trans \in \Re \trans$ and $\bc = (c_1, \ldots, c_n)\trans \in \Re^n$ to minimize the function in \eqref{eq-target}. As a result,
\begin{align*}
\int_{\mI} x_i(t) \beta(t) dt = ~& d_1\int_{\mI}x_i(t)dt + d_2\int_{\mI}x_i(t)(t - 0.5)dt  \\
& + \sum_{j = 1}^n  c_j \int_{\mI} \int_{\mI} x_i(t) K(t, s) x_j(s)dsdt.
\end{align*}
For the penalty term, we have $J(\beta) = \bc\trans \bR \bc$, where $\bR$ is an $n \times n$ matrix with $(i, j)$th entry $r_{ij} = \int_{\mI} \int_{\mI} x_i(t) K(t, s) x_j(s)dsdt$. Denote by $\bS$ an $n \times 2$ matrix with the $(i, j)$th entry
$s_{ij} =  \int_{\mI}x_i(t) \psi_j(t) dt$ for $j = 1, 2$. Let $\bS_i$ and $\bR_i$ denote the $i$th row of $\bS$ and $\bR$, respectively. Now the infinite-dimensional minimization of \eqref{eq-target} becomes the following finite dimensional minimization problem:
\begin{equation}
D(\alpha, \bd, \bc) := n^{-1} \sum_{i = 1}^n V_q\{y_i(\alpha + \bS_i\trans \bd + \bR_i \trans \bc)\} + \lambda \bc \trans \bR \bc.
\label{eq-finite}
\end{equation}

To find the minimizer of $D$, we implement the majorization-minimization (MM) principal. The basic idea is as follows. We firstly look for a majorization function $M(\bth | \bth^{\prime})$, where $\bth = (\alpha, \bd^T, \bc^{T})^{T}$ in this problem, of the target function $D(\bth)$. This majorization function satisfies $D(\bth) < M(\bth | \bth^{\prime})$ for any $\bth \neq \bth^{\prime}$ and $D(\bth) = M(\bth | \bth^{\prime})$ if $\bth = \bth^{\prime}$. Additionally, it should be easy to find the minimizer of $M(\bth | \bth^{\prime})$ for any given $\bth^{\prime}$. Then given an initial value of $\bth$, say $\bth^{(0)}$, we are able to generate a sequence of $\bth$'s, say $\{\bth^{(k)}\}_{k = 1}^{\infty}$, which are defined by $\bth^{(k + 1)} = \argmin_{\bth} M(\bth | \bth^{(k)})$, $k \geq 0$. As long as this sequence converges, the limit is regarded as the minimizer of the objective function $D$.

Given $\bar{\bth} = (\bar{\alpha}, \bar{\bd}\trans, \bar{\bc}\trans)\trans$, let $\br = (r_1, \ldots, r_n)\trans$ with $r_i = y_iV_q^{\prime}\{y_i(\bar{\alpha} + \bS_i\trans\bar{\bd} + \bR_i\trans\bar{\bc})\}/n$ and $$
\bA_{q,\lambda} = \begin{pmatrix}
n  &  \bone_n\trans\bS &  \bone_n\trans\bR\\
\bS\trans\bone_n & \bS\trans\bS & \bS\trans\bR \\
\bR\bone_n & \bR\bS  & \bR\bR + \frac{2nq\lambda}{(q+1)^2}\bR
\end{pmatrix},
$$
where $\bone_n$ denotes a vector of length $n$ with each component equal to 1. 
According to Lemma 2 of \cite{wang2018another}, we can take the majorization function of $D$ as
\begin{align*}
M(\bth | \bar{\bth}) & = \frac{1}{n}\sum_{i = 1}^n V_q\{y_i(\bar{\alpha} + \bS_i\trans\bar{\bd} + \bR_i\trans\bar{\bc})\} + \lambda \bar{\bc}\trans\bR\bar{\bc} \\
& ~~~ + \begin{pmatrix}
\bone\trans\br \\
\bS\trans\br \\
\bR\br + 2\lambda\bR\bar{\bc}
\end{pmatrix}\trans \begin{pmatrix}
\alpha - \bar{\alpha} \\
\bd - \bar{\bd} \\
\bc - \bar{\bc}
\end{pmatrix} + \frac{(q+1)^2}{2nq}\begin{pmatrix}
\alpha - \bar{\alpha} \\
\bd - \bar{\bd} \\
\bd - \bar{\bc}
\end{pmatrix}^{\top} \bA_{q,\lambda} \begin{pmatrix}
\alpha - \bar{\alpha} \\
\bd - \bar{\bd} \\
\bc - \bar{\bc}
\end{pmatrix}.
\end{align*}
It is trivial to show that the minimizer of $M(\bth | \bar{\bth})$ is
$$
\begin{pmatrix}
\bar{\alpha} \\
\bar{\bd} \\
\bar{\bc}
\end{pmatrix} - \frac{nq}{(q+1)^2}\bA_{q,\lambda}^{-1}\begin{pmatrix}
\bone_n\trans\br \\
\bS\trans\br \\
\bR\br + 2\lambda\bR\bar{\bc}
\end{pmatrix}
$$
Then the algorithm proceeds until the sequence of minimizers converges. The limit of this sequence is denoted by $(\halp, \hat{\bd}, \hat{\bc})\trans$, and thus $\hat{\beta}(t) = \hat{d}_1 + \hat{d}_2(t - 0.5) + \sum_{i = 1}^n \hat{c}_i (Kx_i)(t)$. The functional DWD classifier assigns  1 or -1 to a new functional observation $x$ according to whether the statistic $\halp + \int_{\mI} x(t)\hat{\beta}(t)dt$ is positive or negative.

\subsection{Functional DWD with scalar covariates}
The algorithm presented above is to address binary classification problems for univariate functional data. This idea can be further extended to accommodate binary classification when both a functional covariate and finite dimensional scalar covariates are involved. In particular, the training sample consists of $(x_i, y_i, \bz_i), i = 1, \ldots, n$, where $\bz_i = (z_{i1}, \ldots, z_{ip})\trans$ denotes the $p$ dimensional scalar covariates of the $i$th subject. With slight abuse of notation, we consider the following extension of \eqref{eq-target}: 
\begin{equation}
	Q(\alpha, \beta, \bgam) = n^{-1} \sum_{i = 1}^n V_q\left\{y_i\left(\alpha + \int_{\mI} x_i(t) \beta(t) dt + \bz_i\trans \bgam\right)\right\} + \lambda J(\beta),
	\label{eq-target2}
\end{equation}
to build a partial linear DWD classifier.

To solve the minimization problem of \eqref{eq-target2}, we resort to the specific representation of $\hat{\beta}$ in Theorem \ref{th-rep}. Actually it is straightforward to verify that this result still holds in the context of \eqref{eq-target2}. As a result, the infinite dimensional minimization problem of \eqref{eq-target2} is converted to the following finite one:
\begin{equation}
D(\alpha, \bd, \bc, \bgam) := n^{-1} \sum_{i = 1}^n V_q\{y_i(\alpha + \bS_i\trans \bd + \bR_i\trans \bc + \bz_i\trans \bgam)\} + \lambda \bc\trans \bR \bc.
\label{eq-finite2}
\end{equation}
With some modifications, we employ the MM principal to address the minimization problem above. In particular, the majority function of $D$ is taken as
\begin{align*}
M(\bth | \bar{\bth}) & = \frac{1}{n}\sum_{i = 1}^n V_q\{y_i(\bar{\alpha} + \bS_i\trans\bar{\bd} + \bR_i\trans\bar{\bc} + \bz_i\trans\bar{\bgam})\} + \lambda \bar{\bc}\trans\bR\bar{\bc} \\
& ~~~ + \begin{pmatrix}
\bone_n\trans\br \\
\bS\trans\br \\
\bZ\trans\br \\
\bR\br + 2\lambda\bR\bar{\bc} 
\end{pmatrix}\trans \begin{pmatrix}
\alpha - \bar{\alpha} \\
\bd - \bar{\bd} \\
 \bgam - \bar{\bgam} \\
 \bc - \bar{\bc}
\end{pmatrix} + \frac{(q+1)^2}{2nq}\begin{pmatrix}
\alpha - \bar{\alpha} \\
\bd - \bar{\bd} \\
\bgam - \bar{\bgam} \\
\bc - \bar{\bc}
\end{pmatrix}\trans \bA_{q,\lambda} \begin{pmatrix}
\alpha - \bar{\alpha} \\
\bd - \bar{\bd} \\
\bgam - \bar{\bgam} \\
\bc - \bar{\bc}
\end{pmatrix}, 
\end{align*}
where $\br = (r_1, \ldots, r_n)\trans$ with $r_i = y_iV_q^{\prime}\{y_i(\bar{\alpha} + \bS_i\trans\bar{\bd} + \bR_i\trans\bar{\bc} + \bz_i\trans\bar{\bgam})\}/n$, $\bZ = (\bz_1, \ldots, \bz_n) \trans\in \Re^{n \times p}$ and $\bar{\bth} = (\bar{\alpha}, \bar{\bd}\trans, \bar{\bc}\trans, \bar{\bgam}\trans)\trans$, and $$
\bA_{q,\lambda} = \begin{pmatrix}
n  &  \bone_n\trans\bS &  \bone_n\trans\bZ & \bone_n\trans\bR\\
\bS\trans\bone_n & \bS\trans\bS & \bS\trans\bZ & \bS\trans\bR\\
\bZ\trans\bone_n & \bZ\trans\bS & \bZ\trans\bZ & \bZ\trans\bR  \\
\bR\bone_n & \bR\bS  & \bR\bZ & \bR\bR + \frac{2nq\lambda}{(q+1)^2}\bR 
\end{pmatrix}.
$$
Thus the minimizer of $M(\bth | \bar{\bth})$ is given by
$$
\begin{pmatrix}
\bar{\alpha} \\
\bar{\bd} \\
\bar{\bgam} \\
\bar{\bc}
\end{pmatrix} - \frac{nq}{(q+1)^2}\bA_{q,\lambda}^{-1}\begin{pmatrix}
\bone_n\trans\br \\
\bS\trans\br \\
\bZ\trans\br \\
\bR\br + 2\lambda\bR\bar{\bc}.
\end{pmatrix}
$$
We then follow steps in Section \ref{sec-est} to implement classifications on subjects with both a functional covariate and several scalar covariates.

\subsection{Tuning parameter selection}
Here we focus on binary classification when both a functional covariate and several scalar covariates are involved. 
The prediction performance of the proposed classifier depends on the choice of the two tuning parameters $q$ and $\lambda$.
Computing the inverse of the matrix $\bA_{q,\lambda} \in \Re^{(n+3 +p) \times (n + 3 + p)}$ for every combination of $q$ and $\lambda$ would be computationally intensive, especially when sample size is large. 
Instead, we come up with a solution with a lower computational cost to tackle this problem. 
To facilitate the selection procedure, the essential idea in our implementations is to avoid directly computing the inverse matrix of $\bA_{q, \lambda}$ for each combination of $(q, \lambda)$ since it would require substantial time. Write $\bA_{q, \lambda}$ as $$
\bA_{q, \lambda} = \left[
\begin{array}{ccc|c}
n  &  \bone_n\trans\bS &  \bone_n\trans\bZ & \bone_n\trans\bR\\
\bS\trans\bone_n & \bS\trans\bS & \bS\trans\bZ & \bS\trans\bR \\ 
\bZ\trans\bone_n & \bZ\trans\bS & \bZ\trans\bZ & \bZ\trans\bR \\ 
\hline
\bR\bone_n & \bR\bS  & \bR\bZ &  \bR\bR + \frac{2nq\lambda}{(q+1)^2}\bR
\end{array}
\right] := \begin{pmatrix}
	\bB & \bC^T \\
	\bC & \bD
\end{pmatrix}.$$
Therefore, the inverse of $\bA_{q,\lambda}$ admits
\begin{equation}
\bA_{q, \lambda}^{-1} = \begin{pmatrix}
\bB^{-1} + \bB^{-1}\bC\trans(\bD - \bC\bB^{-1}\bC\trans)^{-1}\bC\bB^{-1} & -\bB^{-1}\bC\trans(\bD - \bC\bB^{-1}\bC\trans)^{-1} \\
-(\bD - \bC\bB^{-1}\bC\trans)^{-1}\bC\bB^{-1}& (\bD - \bC\bB^{-1}\bC\trans)^{-1}
\end{pmatrix}
\label{eq-inverseA}
\end{equation}
Note that among these matrices only $D$ depends on $q$ and $\lambda$. The inverse of $\bD - \bC\bB^{-1}\bC^T$ is available from the Sherman-Morrison-Woodbury formula:
\begin{equation}
(\bD - \bC\bB^{-1}\bC^T)^{-1} = \bD^{-1} + \bD^{-1}\bC(B - \bC^T\bD^{-1}\bC)^{-1}\bC^T\bD^{-1}.
\label{eq-inverseD}
\end{equation}
To compute the matrix in \eqref{eq-inverseD}, we need to find the inverse of $\bD$ first.
Let $\bQ\bLam\bQ^{T}$ denote the eigen-decomposition of $\bR$; hence it does not depend on $\lambda$. Then we compute the inverse of $\Pi_{q, \lambda} = \bLam\bLam + \{2nq\lambda/(q+1)^2\}\bLam$ for each $q$ and $\lambda$; it is actually a diagonal matrix. Hence the inverse of $\Pi_{q, \lambda}$ is immediately available for each combination of $q$ and $\lambda$.
Furthermore, $\bD^{-1} = \bQ\Pi_{q, \lambda}^{-1}\bQ^T$ and note that $B - \bC^T\bD^{-1}\bC$ is a $(3 + p) \times (3 + p)$ matrix.
These suggest that it is efficient to compute the (inverse) matrix in \eqref{eq-inverseD}, as long as $p$ is relatively small.

Finally, we employ the expression of $\bA_{q, \lambda}$ in \eqref{eq-inverseA} to compute $\bA_{q, \lambda}^{-1}\begin{pmatrix}
(\bone_n, \bS, \bZ)\trans\br \\
\bR\br + 2\lambda\bR\bar{\bc}
\end{pmatrix}
$
directly. Denote by $\bP$ the inverse of $\bD - \bC\bB^{-1}\bC\trans$. 
By equation \eqref{eq-inverseA}, we have
$$
\bA_{q, \lambda}^{-1}\begin{pmatrix}
(\bone_n, \bS, \bZ)\trans\br \\
\bR\br + 2\lambda\bR\bar{\bc}
\end{pmatrix} = \begin{pmatrix}
\bB^{-1} + \bB^{-1}\bC\trans\bP\bC\bB^{-1} &  -\bB^{-1}\bC\trans\bP\\
-\bP\bC\bB^{-1}& \bP
\end{pmatrix} \begin{pmatrix}
(\bone_n, \bS, \bZ)\trans\br \\
\bR\br + 2\lambda\bR\bar{\bc}
\end{pmatrix}.
$$
With the procedures above, we are able to compute the minimizer of the majorization function $M(\bth | \bar{\bth})$  for different values of $q$ and $\lambda$, and thus solve the minimization problem of \eqref{eq-finite2} efficiently. We employ cross validation to choose the optimal combination of $q$ and $\lambda$ in the following numerical studies.

\section{Theoretical properties} \label{sec-theory}

Let $\fs$ denote the Bayes classifier, which can minimize the probability of misclassification, $P(\fs(X) \neq Y)$.
It is trivial that $\fs(X) = 2\mathbbm{1}\{P(Y = 1 | X = x) > 0.5\} - 1$ a.s. on the set $\{P(Y = 1 | X = x) \neq 0.5\}$.
Given a loss function $\ell$, the associated risk function for a classifier  $f$ is then defined by
$$
L(f, \fs) := \E[l(f) - l(\fs)].
$$
\cite{blanchard2008statistical} established a non-asymptotic bound on $L(\hat{f}, \fs)$ when $\hat{f}$ is the estimated support vector classifier
from a training sample in which each subject consists of multiple scalar covariates and a binary outcome, and $l$ is the corresponding hinge loss function.

Denote by $\hat{f}(x) = \halp + \int_{\mI} x(t)\hat{\beta}(t)dt$ the estimated functional DWD classifier, where $\halp$ and $\hat{\beta}$ minimize the target function $Q$ in \eqref{eq-target}. In the context of functional classification, assuming that the functional covariate $X \in \sX$, the Bayes classifier $\fs$ is a measurable functional from $\sX$ to $\{-1, 1\}$.
Here we aim to establish a non-asymptotic bound on $L(\hat{f}, \fs)$, where $l = V_q$ is the loss function for the functional DWD classifier.

The following conditions are required where (C1) is with regard
to bounds on the noise and (C2) is with regard to bounds on the
kernel and functions.
\begin{enumerate}
  \item[(C1)]
  For $\eta(x) = \prob{Y=1|X=x}$, we require
  $\abs{\eta(x)-1/2-c_q/4}\ge \eta_0>0$ for all $x$ and $q\ge q_0$, 
  thus $q$ large enough so that $c_q < 4 - 2/\eta_0$.
  We also require
  $\min\{\eta(x),1-\eta(x)\}\ge\eta_1>0$ for all $x$, which
  bounds the probabilty away from 0 and 1 and also 
  implies that $\eta_1<1/2$.
  \item[(C2)]
  There exist positive constants $M,A$ such that 
  $\norm{x}_{L^2}\le A$ for $x(t)\in L^2(I)$ and
  $\norm{k}_{L^2}\le M$ for $k$ being 
  the reproducting kernel.
\end{enumerate}

There are two settings to consider similar to \cite{blanchard2008statistical},
which affect how the penalization parameter is controlled.
In setting (S1), the risk is considered via the spectral
properties of the reproducing kernel.  Specifically, the 
penalization parameter $\lmb_n$ is controlled by the
tail sum of the eigenvalues of the reproducing kernel.
In setting (S2), the risk is considered via covering numbers
under the sup-norm.  In contrast, $\lmb_n$ is controlled
via $H_\infty$, the supremum norm $\veps$-entropy.
This control is encapsulated in the $\gamma(n)$ term
defined in the following theorem.

\begin{theorem}
    \label{thm:empRisk}
	Under conditions (C1) and (C2),
	let the penalization parameter $\lambda_n$
	be bounded as 
	$$
	\lambda_n \ge C\left( \gamma(n) + 
	\frac{\log(\delta^{-1}\log n)\vee 1}{n} \right)
	$$
	for some universal constant $C$ with $\gamma(n)=\frac{1}{\sqrt{n}}\inf_{d\in\natural}[
	\frac{2d}{\sqrt{n}} + \frac{A\eta_1}{M}\sqrt{\sum_{j>d}\lmb_j} 
	]$ under (S1) and $\gamma(n)=(x^*)^2/M^2$ under (S2)
	where $x^*$ is the solution to 
	$
	  \int_{0}^{x}
	  \sqrt{ H_\infty(B,\veps) }d\veps
	  = \sqrt{n}x^2/AM
	$ with $B = \{ \beta\in\mathcal{W}_2^m(I)\,:\, \norm{\beta}\le 1\}$.
	For an iid sample of functional-binary pairs 
	$(x_i(t),y_i)$ and FDWD loss function $V_q(\cdot)$,
	let the regularized estimator $\hat{\beta}$ be the 
	solution to 
	$$
	\hat{\beta} = \arg\min_{\beta} n^{-1} 
	\sum_{i = 1}^n V_q\left\{y_i\int_{\mI} x_i(t) 
	\beta(t) dt\right\} + \lambda AMJ(\beta)
	$$ 
	for corresponding classifier $\hat{f}$.
	Then, for $f^*$ being the Bayes classifier
	and $f_\beta$ the classifier corresponding to 
	any arbitrary $\beta$, 
	the following holds with probabiliy at least $1-\delta$,
	$$
	L(\hat{f},f^*) \le
	2\inf_{\beta\in \mathcal{W}_2^m}\left\{
	L(f_\beta,f^*) + 2\lmb_n AMJ(\beta)
	+ \lmb_nk_1 + k_0 )
	\right\}
	$$
	for positive constants $k_0,k_1$.
\end{theorem}

Proof of Theorem~\ref{thm:empRisk} can be found in the appendix.
We extend this theorem to the functional DWD estimator with 
scalar covariates in the following corollary.
For this extension, we require an additional condition:
\begin{enumerate}
  \item[(C3)]
  There exist a positive constant $\Pi$ such that 
  $\norm{z}_{\ell^2}\le \Pi$ for $z\in \real^p$.
\end{enumerate}
The proof of the below corollary follows from that for 
Theorem~\ref{thm:empRisk}.  Instead of considering 
suprema over the ball 
$B(R) = \{ \beta\in\mathcal{W}_2^m(I)\,:\, \norm{\beta}\le R\}$,
we instead consider the product ball
$\mathbb{B}(R) = B(R) \times \{ \gamma\in\real^p\,:\, \norm{\gamma}\le R\}$.

\begin{corollary}
	\label{cor:empRisk}
	Under conditions (C1), (C2), and (C3),
	let $A^* = A + \Pi/M$ be a positive constant, and
	let the penalization parameter $\lambda_n$
	be bounded as 
	$$
	\lambda_n \ge C\left( \gamma(n) + 
	\frac{\log(\delta^{-1}\log n)\vee 1}{n} \right)
	$$
	for some universal constant $C$ with $\gamma(n)=\frac{1}{\sqrt{n}}\inf_{d\in\natural}[
	\frac{2d}{\sqrt{n}} + \frac{A^*\eta_1}{M}\sqrt{\sum_{j>d}\lmb_j} 
	]$ under (S1) and $\gamma(n)=(x^*)^2/M^2$ under (S2)
	where $x^*$ is the solution to 
	$
	\int_{0}^{x}
	\sqrt{ H_\infty(\mathbb{B},\veps) }d\veps
	= \sqrt{n}x^2/A^*M
	$ with $\mathbb{B} = 
	\{ \beta\in\mathcal{W}_2^m(I)\,:\, \norm{\beta}\le 1\}\times
	\{ \gamma\in\real^p\,:\, \norm{\gamma}\le 1\}$.
	For an iid sample of functional-covariate-binary triples 
	$(x_i(t),z_i,y_i)$ and FDWD loss function $V_q(\cdot)$,
	let the regularized estimator $\hat{\beta}$ and $\hat{\gamma}$ 
	be the solution to 
	$$
	(\hat{\beta},\hat{\gamma}) = \arg\min_{\beta,\gamma} n^{-1} 
	\sum_{i = 1}^n V_q\left\{y_i\int_{\mI} x_i(t) 
	\beta(t) dt + z\trans\gamma \right\} + \lambda A^*MJ(\beta)
	$$ 
	for corresponding classifier $\hat{f}$.
	Then, for $f^*$ being the Bayes classifier
	and $f_{\beta,\gamma}$ the classifier corresponding to 
	any arbitrary $\beta$ and $\gamma$, 
	the following holds with probabiliy at least $1-\delta$,
	$$
	L(\hat{f},f^*) \le
	2\inf_{\beta\in \mathcal{W}_2^m,\gamma\in\real^p}\left\{
	L(f_{\beta,\gamma},f^*) + 2\lmb_n A^*MJ(\beta)
	+ \lmb_nk_1 + k_0 )
	\right\}
	$$
	for positive constants $k_0,k_1$.
\end{corollary}

\section{Simulation studies}

\label{sec-simul}
In this section, we considered two different simulation settings to investigate finite sample performance of the proposed classifier. In both settings, the functional covariate was generated in the following way: $X_i(t) = \sum_{j = 1}^{50} \xi_{ij} \zeta_j \phi_j(t)$, where $\xi_{ij}$'s are independently drawn from a uniform distribution on ($-\sqrt{3}, \sqrt{3})$, $\zeta_j = (-1)^{j+1}j^{-1}, j = 1, \ldots, 50$, and $\phi_1(t) = 1$ and $\phi_j(t) = \sqrt{2} \cos ((j-1)\pi t),~ j \geq 2$ for  $t \in [0, 1]$. Observations at 50 times points on the interval [0, 1] were available in each sample path of $X(t)$.  Two scalar covariates ($\bz = (z_1, z_2)\trans$) were independently generated from a truncated normal distribution within the interval $(-2, 2)$ with mean 0 and variance 1. Then the binary response variable $y$ with values $1$ or $-1$ was generated from the logistic model: 
$$
f(X_i, \bz_i) = \alpha_0 + \int_0^1 X_i(t)\beta(t)dt + \bz_i\trans{\bgam},~~p(Y_i = 1) = \frac{\exp\{f(X_i, \bz_i)\}}{1 + \exp\{f(X_i, \bz_i)\}},
$$
where $\alpha_0 = 0.1$ and $f(X, \bz)$ is referred to as the discriminant function in this article.

In the first scenario, the slope function of $X(t)$ was $\beta (t) = e^{-t}$, and $\bgam = (-0.5, 1)\trans$ or $(0, 0)\trans$ to ensure the discriminant function $f$ depends on or not on the scalar covariates, respectively. In the second scenario, the slope function was written as a linear combination of the functional principal components of $X$. Particularly, $\beta(t) = \sum_{j = 1}^{50} 4(-1)^{j+1}j^{-2}\phi_j(t)$, and the coefficient vector of the scalar covariates was $\bgam = (-2, 3)\trans$ or $(0, 0)\trans$. In each simulation scenario, $n = 100$ or 200 curves were generated for training. Then 500 samples were generated as the test set to assess prediction accuracy.

In addition to the proposed functional DWD classifier,  
we also considered several other commonly used functional data classifiers for comparison. The centroid classifier by \cite{delaigle2012achieving} firstly projects each functional covariate onto one specific direction and then performs classification based on the distance to the centroid in the projected space. The functional quadratic discriminant in \cite{galeano2015mahalanobis}  conducts a quadratic discriminant analysis on FPC scores, while the functional logistic classifier fits a logistic regression model on them. Note that the aforementioned classifiers, except our proposed functional DWD with scalar covariates, only account for functional covariates in classification. To study the effect of scalar covariates on classification, we fitted an SVM classifier with only these two scalar covariates when they are involved in the discriminant function. 
In each simulation trial, we randomly generated a training set of size $n = 100$ or 200 to fit all classifiers and then evaluated the predictive accuracy for all of them on a test sample of size 500. To assess the uncertainty in estimating prediction accuracy of each classifier, $B = 500$ independent simulation trials were conducted in each scenario.

\begin{table}
	\caption{ 
		\label{tab-simul}
		The mean misclassification error rates (\%) on the test sample across $M = 500$ simulations with the standard 
		errors (\%) in parentheses. Centroid, the centroid classifier proposed by Delaigle and Hall (2012). Quadratic, the functional quadratic classifier proposed by Galeano et al. (2015). Logistic, the functional logistic classifier. fDWD, our proposed functional DWD method without scalar covariates, while PLfDWD denotes our proposed functional DWD method with scalar covariates. KNN, the k-nearest neighbour classifier introduced in Chapter 8 of \cite{ferraty2006}. S-SVM, an SVM classifier depends only on scalar covariates. The column $\bz$ indicates whether or not the true discriminant function depends on the scalar covariates. }  
  \centering
     \small
     \addtolength{\tabcolsep}{-3pt}   
		\begin{tabular}{l*{9}{c}}
			\multicolumn{9}{c}{Scenario 1} \\\\\hline
			$n$ & $\bz$ & Centroid & Quadratic & Logistic & KNN & fDWD & PLfDWD &  S-SVM  \\
			\hline
			100 & Yes  & 41.3 (3.7)& 43.2 (3.0) & 41.9 (3.1)& 43.4 (3.3) & 39.8 (2.7)  &  \textbf{32.3 (2.7)} & 37.9 (5.1)  \\
			& No & 39.1 (3.7) & 40.7 (3.3) & 39.5 (3.1)& 40.8 (3.6) & \textbf{37.7 (2.8)}  & &  \\
			\hline
			200 & Yes & 40.2 (2.8) & 41.5 (2.8) & 40.5 (2.6) & 43.0 (2.8) & 39.3 (2.4) & \textbf{31.3 (2.4)} & 35.5 (3.0)\\
			& No & 37.9 (2.7) & 39.1 (2.8) & 38.2 (2.4) & 40.6 (3.1) & \textbf{37.1 (2.3)} & & \\
			\hline\\
			\multicolumn{9}{c}{Scenario 2} \\\\\hline
			$n$ & $\bz$ & Centroid & Quadratic & Logistic & KNN & fDWD & PLfDWD & S-SVM \\
			\hline
			100 & Yes& 21.7 (2.2) & 22.4 (2.4) & 21.9 (2.2) & 22.6 (2.5) & 21.0 (2.0 & \textbf{11.0 (1.9)} & 35.7 (4.5) \\
			& No & 11.1 (1.6) & 11.4 (1.6) & 11.4 (1.7) & 11.5 (1.6) & \textbf{10.4 (1.4)} & & \\
			\hline
			200 & Yes & 21.3 (2.0) & 21.5 (2.0) & 21.2 (1.9) & 22.2 (2.2) & 20.8 (1.8) & \textbf{10.5 (1.9)} & 33.9 (2.9)  \\
			& No& 10.7 (1.5) & 10.7 (1.4) & 10.6 (1.4) & 11.1 (1.5) & \textbf{10.2 (1.4)} & & \\
			\hline
		\end{tabular}
	\addtolength{\tabcolsep}{3pt}
\end{table}

Table \ref{tab-simul} summarizes the mean and the standard error of the misclassification error rate of each classifier. In the first scenario, the proposed functional DWD classifier with scalar covariates is considerably more accurate than any other classifier in terms of prediction. This is not surprising since even the SVM classifier with only scalar covariates outperforms the functional classifiers that did not take scalar covariates into consideration. This fact implies the importance of accounting for scalar covariates when the true discriminant function indeed depends on them. Additionally, whether or not the true discriminant function depends on the scalar covariates in these settings, the misclassification error rates of our proposed functional DWD classifier, are very close to the Bayes errors, which are 0.283 and 0.376, respectively. 
As the projection function in the centroid classifier and the slope function in the functional logistic regression by \cite{leng2005classification} are represented in terms of FPCs, these two classifiers should be in favor in the second scenario. This was justified by comparing prediction accuracy of these classifiers. However, our proposed classifier still dominates all competitors no matter whether the true discriminant function depends on the scalar covariates. A plausible reason why the proposed classifier is superior to the centroid and logistic classifiers is that the roughness of the projection direction is appropriately controlled in our method. Once again, the misclassification rates of our proposed classifiers, are very close to the Bayes errors, which are 0.086 with scalar covariates and 0.099 without scalar covariates, respectively. 

\bigskip

\noindent \blue{Pj: You can decide whether boxplots of misclassification error rate of each classifier in each scenario should be put in the Appendix. All of them are now in the folder ``Figure".}
 
\section{Real data examples} \label{sec-real}

In this section,  we apply the proposed classifiers as well as several alternative classifiers to one real-world example to demonstrate the performance of our proposal. 

Alzheimer's disease (AD) is an irreversible and progressive brain disorder that can lead to more and more serious dementia symptoms over a few years.
Previous studies showed that increasing age is one of the most important risk factor of AD, and most patients with AD are above 65. However, there also exist substantial early-onset Alzheimer's whose ages are under 65. Things are even worse considering the fact that there is no current cure for AD and AD would even eventually destroy people's ability to perform the simplest tasks. Due to the reasons above, studies of AD have received considerable attention in the past few years.

In our study, the data
were obtained from the ongoing Alzheimer’s Disease Neuroimaging Initiative (ADNI), which aims to unit researchers from the world to collect, validate and analyze relevant data. In particular, the ADNI is interested in identifying biomarkers of AD from genetic, structural, and functional neuroimaging, and clinical data. 
The dataset consists of two main parts. The
first part is neuroimaging data collected by diffusion tensor imaging (DTI). More specifically, fractional anisotropy (FA) values were measured at 83 locations
along the corpus callosum (CC) fiber tract for each subject. 
The second part is composed of demographic features like gender (a categorical variable),
handedness (left hand or right hand, a categorical variable), the age, the education level,  the AD status and mini-mental state examination (MMSE) scores. The AD status is a categorical variable with three levels: normal control (NC), mild cognitive impairment (MCI)  and Alzheimer’s disease (AD). We combine the first two categories into one single category for simplicity, and then this status variable is treated as a binary outcome in our following analysis. The MMSE is one
of the most widely used test of cognitive functions such as orientation, attention, memory, language and visual-spatial skills  for assessing the level of dementia a patient may have. 
A more detailed description of the data can be found in http://adni.loni.usc.edu. 
Previous studies, such as \cite{li2017} and \cite{tang2018}, focused on building regression models to investigate the relationship between the progression of the AD status and the neuroimaging and demographic data. However, 
our main objective is to use the DTI data and demographic features to predict
the status of AD.

\begin{figure}[H]
\centering
{\includegraphics[width=6cm]{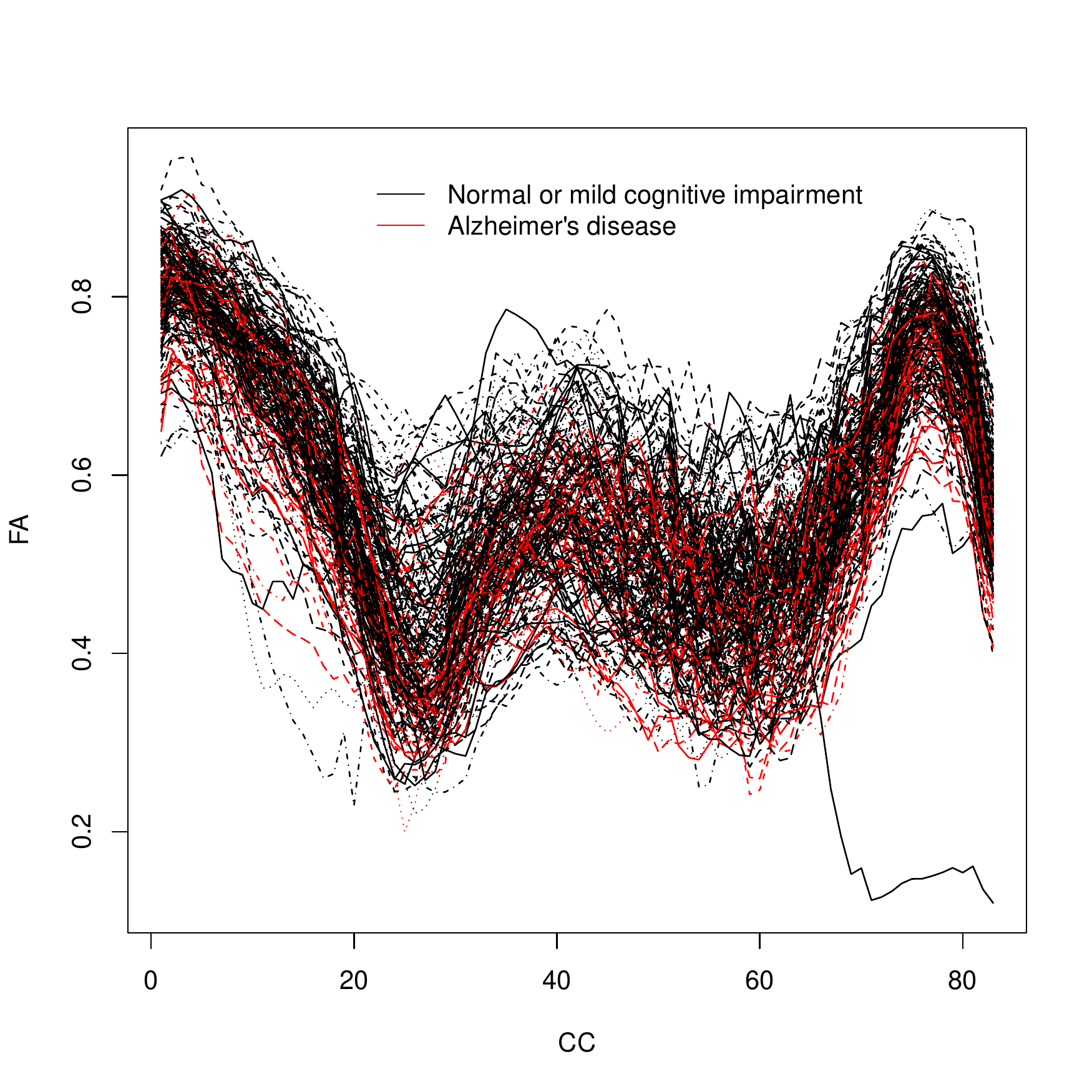}}
\caption{FA profiles of two groups of people: the normal or mild cognitive impairment group and the AD group.}
\label{fig:ADNI}
\end{figure}

We had $N = 214$ subjects in our analysis after removing 3 subjects with missing values. Among them, there are $N_0 = 172$ subjects from the first group, i.e.,
the status of them is either NC or MCI, and $N_1 = 42$ subjects from the AD group. 
The functional predictor $X(t)$ was taken as the FA profiles, which are displayed in Figure \ref{fig:ADNI}, and the scalar covariates $\bz$ consisted of gender, handedness, the age, the education level and the MMSE score. To justify the importance of incorporating FA profiles in classification, we considered an SVM classifier with only these scalar covariates. 
To compare the prediction performance of
each classifier, these 214 subjects were randomly divided into a training set with $n$ subjects and a test set with the other
$N - n$ subjects. In the study, we considered two particular choices of $n$: $0.5N$ and $0.8N$. 
Following this rule, we randomly splitted the whole dataset to the training and test set for $M = 500$ times.

\begin{table}
	\caption{
		\label{tab:ADNI}
		The mean misclassification error rates (\%) across $M = 500$ random splits with 
		standard errors in brackets.}
	\addtolength{\tabcolsep}{-4pt}   
	\centering
	\small
	
	\begin{tabular}{@{}llcccccc@{}}
		\hline
		$n$ & Centroid & Quadratic & Logistic & KNN & fDWD & PLfDWD & S-SVM \\
		107 & 19.6 (0.029)   &  20.6 (0.031) &   21.4 (0.03) &   21.0 (0.028) &    19.8 (0.03) &    11.2 (0.028) &   26.3 (0.032)   \\
		171 &  19.8 (0.054) &    20.5 (0.056) &    21.5 (0.057) &    21.7 (0.054) &    19.8 (0.054) &     9.9 (0.048) &    27.2 (0.059) \\
		\hline
	\end{tabular}
\addtolength{\tabcolsep}{4pt}   
\end{table}

Table \ref{tab:ADNI} summarizes the mean misclassification error rates and the standard errors across the 500 splits for each classifier. 
When scalar covariates are not accounted for in the functional data classification, our proposed method (fDWD) outperforms all other competitors except the centroid method in terms of prediction accuracy. Even more remarkably, incorporating scalar covariates, though in a linear manner, results in a substantial reduction in misclassification errors in our proposed classifier, around half of the errors of all functional classifiers without scalar covariates. You might argue that this occurred because  scalar covariates are highly predictive of the AD status while the functional covariate is not. However, the prediction performance of S-SVM disproved this argument, as the SVM classifier that only considered scalar covariates performed even worse than those with only the functional covariate in terms of prediction accuracy. On the one hand, these comparisons indicate superiority of our proposed classifier. On the other hand, they also suggest accounting for scalar covariates in an appropriate way is able to enhance prediction accuracy when discriminating functional data. 

\section{Conclusion} \label{sec-end}
In this paper we propose a novel methodology that combines the idea of the canonical DWD classifier and regularized functional linear regression under the RKHS framework to classify functional data. The use of RKHS enables us to control roughness of the estimated projection direction, and thus enhances prediction accuracy in comparison with conventional functional logistic regression and the centroid classifier. Moreover, we further extend the framework to classifying subjects with both a functional covariate and several scalar covariates. Though we focus on the specific loss function to achieve nice properties of DWD classifiers in our study, this framework can be extended to other loss functions such as the logistic loss function in the functional logistic regression and the hinge loss function in functional SVM classifiers. Moreover, the scalar covariates are incorporated in our classifier in a linear manner to achieve a good trade off between flexibility and interpretability. However, nonlinear or nonparametric forms of scalar covariates can also be incorporated into our current framework, as long as we adopt appropriate regularizations to avoid overfitting. This direction deserves future investigation in both theory and practice.

Numerical studies including both simulation studies and one real-world application suggest that the proposed classifier is superior to many other competitors in terms of predication accuracy. The application of our classifier to a study of Alzheimer's disease provides numerical evidence that both neuroimaging data and demographic features are relevant to AD, and ignoring either of them would deteriorate prediction accuracy of the AD status.

\bibliography{dwd}{}
\bibliographystyle{rss}

\section*{Appendix}

\subsection*{A.1: Proof of Theorem 1}

Based on the observations above, we have $Kx_i \in \mH_1, i = 1, \ldots, n$. Therefore,
the solution to \eqref{eq-target} can be written as
	$$
	\hbeta_n(t) = \sum_{l = 1}^N d_l\psi_l(t) + \sum_{i = 1}^n c_i(Kx_i)(t) + \rho(t),
	$$
where $\rho(t)$ is the orthogonal complement of $\sum_{i = 1}^n c_i(Kx_i)(t)$ in $\mH_1$. To prove \eqref{eq-sol}, we just need to check that $\rho(t) = 0$.
	
Let $\bpsi(t) = (\psi_1(t), \ldots, \psi_N(t))^T$.
Plugging the solution into \eqref{eq-target}, we have
	$$
	\min_{\alpha, \bd, \bc, \rho} n^{-1} \sum_{i = 1}^n V_q\left\{y_i\left(\alpha + \int_{\mI} x_i(t) \left(\bd^T\bpsi(t) + \sum_{j = 1}^n c_j (Kx_j)(t) + \rho(t) \right) dt\right)\right\} + \lambda J(\beta).
	$$
In the first term,
	\begin{align*}
	\int_{\mI} x_i(t) & \left(\bd^T\bpsi(t) + \sum_{j = 1}^n c_j (Kx_j)(t) + \rho(t) \right) dt\\ 
	& = \int_{\mI} x_i(t) \left(\bd^T\bpsi(t) + \sum_{j = 1}^n c_j (Kx_j)(t)\right) dt +  \int_{\mI} x_i(t) \rho(t) dt \\
	& = \int_{\mI} x_i(t) \left(\bd^T\bpsi(t) + \sum_{j = 1}^n c_j (Kx_j)(t)\right) dt  + \langle Kx_i, \rho \rangle_{\mH} \\
	& = \int_{\mI} x_i(t) \left(\bd^T\bpsi(t) + \sum_{j = 1}^n c_j (Kx_j)(t)\right) dt.
	\end{align*}
In other words, the first term does not depend on $\rho$. As we know, the second term, $\lambda J(\beta)$, is minimized when $\rho(t) = 0$ since $\rho$ is orthogonal to $\sum_{l = 1}^N d_l\psi_l(t) + \sum_{i = 1}^n c_i(Kx_i)(t)$ in $\mH_1$.

\subsection*{A.2: Proofs of Theorem~\ref{thm:empRisk}}

To prove Theorem~\ref{thm:empRisk}, we first prove the
following lemmas.
We also define
$B(R) = \{ \beta\in\mathcal{W}_2^m(I)\,:\, \norm{\beta}\le R\}$
with $R\in\mathcal{R}\subset\real^+$ a countable set.
In what follows, Lemmas~\ref{lem:varS1} and~\ref{lem:empS1}
are for setting (S1) and Lemmas~\ref{lem:varS2}
and~\ref{lem:empS2} are for setting (S2).

\begin{lemma}
  Let $f = \int_\mathcal{I}x_i(t)\beta(t)dt$ with 
  $x(t)\in L^2(I)$ such that $\norm{x(t)}_{L^2} \le A\in\real^+$
  and reproducing kernel $k:I\times I\rightarrow\real$ 
  such that $\norm{k}_{L^2}\le M\in\real^+$.
  Then, 
  $$
    \norm{V_q(yf)}_{\infty,R} = \sup_{\beta\in B(R)} V_q(yf)
     \le 1+RAM~~\text{a.s.}
  $$ 
  for all $m\in\mathcal{M}$.
\end{lemma}
\begin{proof}
  Via the Riesz representation theorem, we have 
  $f(\cdot) = \iprod{\beta}{\cdot}$ for some 
  $\beta\in\mathcal{W}_2^m(\mathcal{I})$.  
  Using \holders inequality or the Cauchy-Schwarz 
  inequality, we have the following bound.
  \begin{align*}
  \abs*{V_q\left( yf \right)}
  &\le 
  1 + \abs*{\int_I x(t)\beta(t) dt }\\
  &=
  1 + \abs*{\int_{I} x(t) \iprod{\beta}{k(t,)} dt }\\
  &\le 
  1 + \norm{\beta}\norm{x}_{L^2}\norm{k}_{L^2}\\
  &= 1 + RAM.
  \end{align*}
\end{proof}

\begin{lemma}
  \label{lem:varS1}
  For all $\beta,\beta'\in\mathcal{W}_2^m(\mathcal{I})$
  with $f = \iprod{\beta}{}$ and $f'=\iprod{\beta'}{}$,
  $
  \mathrm{Var}[{ V_q(yf) - V_q(yf') }] \le \xv[(f-f')^2].
  $
  Furthermore, under Conditions (C1) and (C2), 
  $$
    \xv[(f-f')^2] \le \left(
      \frac{RAM}{\eta_1} + \frac{5}{\eta_0}
    \right) L(f,f^*)
  $$ 
  for all $\beta,\beta'\in B(R)$ and all $R\in\mathcal{R}$
  where 
  $f^*\in\argmin_{g\in\mathfrak{G}}\xv[V_q(g)]$ and
  $L(f,f^*) = \xv{[ V_q(yf)-V_q(yf^*) ]}$  is the 
  associated risk.
\end{lemma}

\begin{proof}
  First, we note that for $q\in(0,\infty)$ that $V_q(u)$ 
  is differentiable with $\abs{\frac{d}{du}V_q(u)}\le1$.
  Hence, $V_q$ is Lipschitz implying that 
  $$
  \xv{[ (V_q(yf) - V_q(yf'))^2 ]} 
  \le \xv[(f-f')^2]
  $$
  proving the first part of the lemma.
  \\
  
  For the bound on $\xv[(f-f')^2]$, let
  $f^*(x) = 2\Indc{ \eta(x) > 0.5}-1$
  with $\eta(x)=\prob{Y=1|X=x}$.
  Without loss of generality, we will consider the case
  $\eta>1/2$ and $f^*=1$.
  Note that $V_q(f)\ge (1-f)_+$ for $q\in(0,\infty)$.
  Note also that 
  $f = \int x\beta dt = \iprod{Kx}{\beta}$, so 
  $\sup_x f \le  RAM$ 
  for $\beta\in B(R)$.
  We have the ratio
  \begin{align*}
  \frac{
  	\xv[(f-f^*)^2|X=x]
  }{
  	\xv[V_q(yf)-V_q(yf^*)|X=x]
  }	
  &=
  \frac{
  	(f-1)^2
  }{
  	\eta(V_q(f)-V_q(1)) +
  	(1-\eta)(V_q(-f)-V_q(-1))
  }\\
  &= 
  \frac{
  	(f-1)^2
  }{
  	\eta V_q(f) +
  	(1-\eta)V_q(-f)
  	+\eta(2-c_q)-2
  }\\
  &\le
  \frac{
  	(f-1)^2
  }{
  	\eta (1-f)_+ +
  	(1-\eta)(1+f)_+
  	+\eta(2-c_q)-2
  }
  \end{align*}
  If $f<-1$, then denote $g = -1-f$ so $f = -1-g$. 
  Also, note that $2\eta-1 \ge 2\eta_0$.
  Thus, if we choose $q$ such that 
  $c_q < 4 - 2/\eta_0$, then
  \begin{align*}
  \frac{
  	\xv[(f-f^*)^2|X=x]
  }{
  	\xv[V_q(f)-V_q(f^*)|X=x]
  }	
  &\le
  \frac{
  	(f-1)^2
  }{
  	\eta (1-f)_+ 
  	+\eta(2-c_q)-2
  }\\
  &=
  \frac{
  	(g+2)^2
  }{
  	\eta (g+2) 
  	+\eta(2-c_q)-2
  }\\
  &=
  \frac{
  	(g+2)^2
  }{
  	\eta g 
  	+\eta(4-c_q)-2
  }\\
  &\le
  \frac{
  	g
  }{
  	\eta  
  }
  +
  \frac{
  	4
  }{
  	\eta  
  }
  +
  \frac{
  	1
  }{
  	\eta_0  
  }\\
  &\le 
  2RAM + 5/\eta_0
  \le
  RAM/\eta_1 + 5/\eta_0.
  \end{align*}
\end{proof}

\begin{lemma}
	\label{lem:empS1}
	Under Conditions (C1) and (C2),
	let $\phi_R$ be a sequence of subroot functions 
with $r_R^*$ being the solution to 
$\phi_R(r) = r/C_R$.  Then, for all 
$R\in\mathcal{R}$, $\beta_0\in B(R)$, corresponding 
$f_0=\iprod{\beta_0}{}$, $r\ge r_R^*$, and
$d^2(f,f_0)=\xv[(f-f_0)^2]$
$$
\xv\left[
\sup_{\substack{\beta\in B(R)\\d^2(f,f_0)\le r}}
(P-P_n)(\ell(f)-\ell(f_0))
\right] \le
\frac{4}{\sqrt{n}}
\inf_{d\in\natural} \left(\sqrt{dr}+AR\sqrt{\sum_{j>d}\lmb_j}\right)
=: \phi_R(r)
$$
with 
$$
r^*\le
\frac{8C_R^2}{\sqrt{n}}\inf_{d\in\natural}\left[
\frac{2d}{\sqrt{n}} + \frac{A\eta_1}{M}\sqrt{\sum_{j>d}\lmb_j} 
\right].
$$
\end{lemma}

\begin{proof}
	We first define the Rademacher average for a function
	$g:\mathcal{X}\times\mathcal{Y}\rightarrow\real$ to be 
	$\mathcal{R}_n g = n^{-1}\sum_{i=1}^n\veps_ig(X_i,Y_i)$
	for $\veps_1,\ldots,\veps_n$ iid Rademacher random variables
	that are independent of the $(X_i,Y_i)$.  This can be applied
	to a class of $\mathcal{G}$ by denoting 
	$\mathcal{R}_n\mathcal{G} = \sup_{g\in\mathcal{G}}\mathcal{R}_ng$.
\\
	
	Lemma~{6.7} from \cite{blanchard2008statistical} uses a standard 
	symmetrization trick to prove that for 
	any such collection of real valued functions $\mathcal{G}$,
	some 1-Lipschitz function $\varphi$, and some 
	any $g_0\in\mathcal{G}$ that 
	\begin{equation}
	\label{lem67Blan}
	\xv\left[
	\sup_{g\in\mathcal{G}} 
	\abs{(P-P_n)(\varphi(g)-\varphi(g_0))}
	\right] \le 4 \xv \mathcal{R}_n\left\{
	g-g_0\,:\,g\in\mathcal{G}
	\right\}
	\end{equation}
	Lemma~6.8 from \cite{blanchard2008statistical} comes from 
	\cite{mendelson2003kernel}.  It builds off the previous result 
	to note that 
	\begin{align}
	\nonumber
	\xv \mathcal{R}_n\left\{
	g\in\mathcal{H}_k\,:\,\norm{g}_k\le R,\,
	\norm{g}_{2,p}^2\le r
	\right\} 
	&\le \frac{1}{\sqrt{n}}
	\inf_{d\in\natural} \left(\sqrt{dr}+R\sqrt{\sum_{j>d}\lmb_j}\right)\\
	&
	\label{lem68Blan}
	\le
	\sqrt{\frac{2}{n}}\left(
	\sum_{j=1}^\infty \min\{r,R^2\lmb_j\}
	\right)^{1/2}
	\end{align}
	where $\lmb_j$ is the $j$th eigenvalue of the 
	reproducing kernel and where for 
	$g(t) = \sum_{j=1}^\infty\alpha_j\psi_j(t)$ 
	the norms are $\norm{g}_k^2 = \sum \alpha_j^2$ and
	$\norm{g}_{2,p}^2 = \sum \lmb_j\alpha_j^2$.
\\	
	
	By noting that $V_q(\cdot)$ is a 1-Lipschitz function 
	for any choice of $q\in\real^+$, we apply Equation~\ref{lem67Blan} 
	for $x(t)$ with $\norm{x}_{L^2}\le A$ to get
	\begin{multline*}
	\xv\left[
	\sup_{\substack{\beta\in B(R)\\d^2(f,f_0)\le r}}
	(P-P_n)(V_q(f)-V_q(f_0))
	\right] \\
	\le  4\mathcal\xv{R}_n\left\{
	f - f_0\,:\, 
	\beta\in B(R), d^2(f,f_0)\le r
	\right\}
	\le  4\xv\mathcal{R}_n\left\{
	f - f_0\,:\, 
	\norm{f} \le AR, \norm{f}_{2,P}^2\le r
	\right\}.
	\end{multline*}
	Then, application of Equation~\ref{lem68Blan} gives
	$$
	\xv\left[
	\sup_{\substack{\beta\in B(R)\\d^2(f,f_0)\le r}}
	(P-P_n)(V_q(f)-V_q(f_0))
	\right]
	\le
	\frac{4}{\sqrt{n}}
	\inf_{d\in\natural} \left(\sqrt{dr}+AR\sqrt{\sum_{j>d}\lmb_j}\right)
	=: \phi_R(r).
	$$
	Thus, we aim to solve $r/C = \phi(r)$.  Let $d^*$ be the minimizer
	over $d\in\natural$, which exists due to the reproducing kernel 
	being a trace class operator, which in turn implies that 
	$\sqrt{dr}+AR(\sum_{j>d}\lmb_j)^{1/2}$ is finite for all $d$ 
	and tends to $\infty$ as $d\rightarrow\infty$.  
	Then, application of the quadratic formula and the 
	convexity result $(x+y)^2\le2x^2+2y^2$ gives
	\begin{align*}
	0 &= r^* - 
	\frac{4C}{\sqrt{n}}
	\left(\sqrt{d^*r^*}+AR\sqrt{\sum_{j>d^*}\lmb_j}\right)\\
	r^* &= \left[
	\frac{2C\sqrt{d^*}}{\sqrt{n}} +
	\left\{
	\frac{4C^2d^*}{n} + \frac{4C}{\sqrt{n}}AR\sqrt{\sum_{j>d}\lmb_j}
	\right\}^{1/2}
	\right]^2\\
	&\le
	\frac{8C^2}{n}\left[
	2d^* + \frac{AR\sqrt{n}}{C}\sqrt{\sum_{j>d^*}\lmb_j} 
	\right].
	\end{align*}
	Finally, choosing $C \ge MR/\eta_1$ gives
	$$
	r^*\le
	\frac{8C^2}{\sqrt{n}}\inf_{d\in\natural}\left[
	\frac{2d}{\sqrt{n}} + \frac{A\eta_1}{M}\sqrt{\sum_{j>d}\lmb_j} 
	\right].
	$$
\end{proof}

\begin{lemma}
    \label{lem:varS2}
    For all
	$\beta,\beta'\in\mathcal{W}_2^m(\mathcal{I})$
	with $f = \iprod{\beta}{}$ and $f'=\iprod{\beta'}{}$,
	$
	\mathrm{Var}[{ V_q(yf) - V_q(yf') }] \le \xv[(f-f')^2].
	$
	Furthermore, under Conditions (C1) and (C2), 
	$$
	\xv{[ (V_q(yf) - V_q(yf'))^2 ]} \le \left(
	{RAM} + \frac{1}{\eta_0}
	\right) L(f,f^*)
	$$ 
	for all $\beta,\beta'\in B(R)$ and all $R\in\mathcal{R}$
	where 
	$f^*\in\argmin_{g\in\mathfrak{G}}\xv[V_q(g)]$ and
	$L(f,f^*) = \xv{[ V_q(yf)-V_q(yf^*) ]}$  is the 
	associated risk.
\end{lemma}
\begin{proof}
	By choice of the metric $d(f,f_0)$, we have immediately
	that 
	$
	\xv{[ (V_q(yf) - V_q(yf'))^2 ]} 
	\le \xv[(V_q(f)-V_q(f_0))^2].
	$
	Next, we aim to bound
	$$
	\frac{
		\xv[(V_q(yf)-V_q(yf^*))^2|X=x]
	}{
		\xv[V_q(yf)-V_q(yf^*)|X=x]
	}	
	$$
	Recalling that $f^*(x) = 2\Indc{\eta(x)>0.5}-1$, we
	can assume that $\eta>0.5$ and that $f^*=1$.
	We also note that $V_q(f)\ge (1-f)_+$ for $q\in(0,\infty)$, 
	and that 
	$f = \int x\beta dt = \iprod{Kx}{\beta}$, so 
	$\sup_x f \le  RAM$ 
	for $\beta\in B(R)$.  Thus,
	\begin{align*}
	&\frac{
		\xv[(V_q(yf)-V_q(yf^*))^2|X=x]
	}{
		\xv[V_q(yf)-V_q(yf^*)|X=x]
	}\\	
	&~= 
	\frac{
		\eta[ V_q(f)^2 - 2V_q(f)V_q(1) + V_q(1)^2 ] +
		(1-\eta)[ V_q(-f)^2 - 2V_q(-f)V_q(-1) + V_q(-1)^2 ]
	}{
		\eta(V_q(f)-V_q(1)) +
		(1-\eta)(V_q(-f)-V_q(-1))
	}\\
    &~\le
    \frac{
    	\eta[ V_q(f)^2 - 2c_qV_q(f) + c_q^2 ] +
    	(1-\eta)[ V_q(-f)^2 - 4V_q(-f) + 4 ]
    }{
      \eta (1-f)_+ +
      (1-\eta)(1+f)_+
      +\eta(2-c_q)-2
    } 
	\end{align*}
    For $f \le -1$, then denote $g = -1-f$ so $f = -1-g$. 
    Also, note that $2\eta-1 \ge 2\eta_0$.
    Thus, if we choose $q$ such that 
    $c_q < 4 - 2/\eta_0$, then
    \begin{align*}
    &\frac{
    	\xv[(V_q(yf)-V_q(yf^*))^2|X=x]
    }{
    	\xv[V_q(yf)-V_q(yf^*)|X=x]
    }\\	
    &~\le
    \frac{
    	\eta[ (1-f)^2 - 2c_q(1-f) + c_q^2 ] +
    	(1-\eta)[ c_q^2 + 4 ]
    }{
    	\eta (1-f) +
    	\eta(2-c_q)-2
    } \\	
  &~\le
  \frac{
    \eta[ (2+g)^2 - 2c_q(2+g) ] +
    4(1-\eta) + c_q^2
  }{
    \eta (2+g) +
    \eta(2-c_q)-2
  } \\
  &~\le
  \frac{
  	\eta g^2 + (4-2c_q)\eta g +
  	4(1-c_q) + c_q^2
  }{
  	\eta g +
  	\eta(4-c_q)-2
  } \\
  &~\le
  g + 
  \frac{
  	(2 - \eta c_q) g +
  	4(1-c_q) + c_q^2
  }{
  	\eta g +
  	\eta(4-c_q)-2
  }\\ 
  &~\le RAM + \frac{2}{\eta} - c_q \le RAM + \frac{1}{\eta_0}.
    \end{align*}
\end{proof}

\begin{lemma}
	\label{lem:empS2}
	Under Conditions (C1) and (C2),
	let $\phi_R$ be a sequence of subroot functions 
	with $r_R^*$ being the solution to 
	$\phi_R(r) = r/C_R$.  Then, for all 
	$R\in\mathcal{R}$, $\beta_0\in B(R)$, corresponding 
	$f_0=\iprod{\beta_0}{}$, $r\ge r_R^*$, and
	$d^2(f,f_0)=\xv[(V_q(g)-V_q(g_0))^2]$
	$$
	\xv\left[
	\sup_{\substack{\beta\in B(R)\\d^2(f,f_0)\le r}}
	(P-P_n)(\ell(f)-\ell(f_0))
	\right] \le
	\frac{96RA}{\sqrt{n}}\xi\left(\frac{\sqrt{r}}{4RA}\right)
	+\frac{64R^3A^3M}{nr}\xi\left(\frac{\sqrt{r}}{4RA}\right)^2
	:=  \phi_R(r)
	$$
	with 
	$
	r^*
	\le 625C_R^{*2}(x^*)^2/M^2.
	$
\end{lemma}

\begin{proof}
	Lemma~6.10 from \cite{blanchard2008statistical} states that 
	for $\mathcal{G}$ a separable class of real functions
	in sup-norm such that $\norm{g}_\infty\le M$ and
	$\xval{g^2}\le\sigma^2$, the following bound holds:
	\begin{equation}
	\label{lem610blan}
	\xv\left[
	\sup_{g\in\mathcal{G}} \abs{ (P-P_n)g }
	\right]  \le
	\frac{24}{\sqrt{n}}\int_{0}^{\sigma}\sqrt{ H_\infty(\veps) }d\veps
	+\frac{MH_\infty(\sigma)}{n}
	\end{equation}
	where $H_\infty(\veps)$ is the supremum norm $\veps$-entropy 
	for $\mathcal{G}$.
\\
	
	We first note that 
	$$
	\{
	V_q(f)-V_q(f_0)\,:\,\beta\in B(R),\,d^2(f,f_0)\le r 
	\} \subset 
	\{
	V_q(f)\,:\,\beta\in B(2R),\,\xv[V_q(f)^2]\le r
	\}
	$$
	and secondly note that the DWD loss function is 
	1-Lipschitz, which implies that 
	$$
	\abs{ V_q(f)-V_q(f') } \le
	\abs*{
		\int_I x(t)\beta(t)dt -
		\int_I x'(t)\beta'(t)dt 
	} 
	\le \norm{\beta-\beta'}\norm{x-x'}
	\le 2A\norm{\beta-\beta'}
	$$
	Thus, we can bound 
	$$
	H_\infty(\{
	V_q(f)\,:\,\beta\in B(2R),\,\xv[V_q(f)^2]\le r
	\},\veps) \le
	H_\infty(B(4AR),\veps).
	$$
	and application of Equation~\ref{lem610blan} gives
	\begin{align*}
	\xv&\left[
	\sup_{\substack{\beta\in B(R)\\d^2(f,f_0)\le r}}
	(P-P_n)(V_q(f)-V_q(f_0))
	\right] \\
	&\le
	\frac{24}{\sqrt{n}}\int_{0}^{\sqrt{r}}
	\sqrt{ H_\infty(B(4AR),\veps) }d\veps
	+\frac{4RAMH_\infty(B(4AR),\sqrt{r})}{n}\\
	&\le
	\frac{96RA}{\sqrt{n}}\int_{0}^{\sqrt{r}/4RA}
	\sqrt{ H_\infty(B(1),\veps) }d\veps
	+\frac{4RAMH_\infty(B(1),\sqrt{r}/2R)}{n}\\
	&\le
	\frac{96RA}{\sqrt{n}}\xi\left(\frac{\sqrt{r}}{4RA}\right)
	+\frac{64R^3A^3M}{nr}\xi\left(\frac{\sqrt{r}}{4RA}\right)^2
	:=  \phi_R(r)
	\end{align*}
	where
	$
	\xi(x) = \int_{0}^{x}
	\sqrt{ H_\infty(B(1),\veps) }d\veps
	$.
	This final bound is a sub-root function in terms of $r$.
\\
	
	For $x^*$, the solution to $\xi(x)=\sqrt{n}x^2/AM$, we can bound 
	$r_R^*$, the solution to $\phi_R(r)=r/C_R$ with $C_R>RAM$, as follows.
	First, we note that $\xi(x)/x$ is decreasing.  We also choose a $c>0$
	so that $t_R^* := c^2C^{*2}_R(x^*)^2/A^2M^2$.  Therefore,
	$$
	\xi\left(\frac{\sqrt{t_R^*}}{4RA}\right)
	\le
	\frac{cC_R}{4RAM}\xi(x^*)
	= 
	\frac{cC_R\sqrt{n}(x^*)^2}{4RA^2M^2}
	=
	\frac{\sqrt{n}t_R^*}{4RcC_R}.
	$$
	Thus,
	$$
	\phi_R(t^*_R) \le
	\frac{96RA}{\sqrt{n}}\frac{\sqrt{n}t_R^*}{4RcC_R}
	+\frac{64R^3A^3M}{nt_R^*}\left[\frac{\sqrt{n}t_R^*}{4RcC_R}\right]^2
	= 
	\left[\frac{24A}{c}
	+\frac{4RA^3M}{c^2C_R}\right]\frac{t_R^*}{C_R^*}
	\le 
	\left[\frac{24A}{c}
	+\frac{4A^2}{c^2}\right]\frac{t_R^*}{C_R^*}.
	$$
	Therefore, taking $c>25A$ results in $\phi_R(t_R^*) \le t^*/C_R^*$
	implying that $t_R^*\ge r_R^*$ giving finally that 
	$
	r_R^* \le 625C_R^{*2}(x^*)^2/M^2.
	$
\end{proof}

\begin{proof}[Proof of Theorem~\ref{thm:empRisk}]
	Given the above lemmas, we apply the ``model selection'' 
	Theorem~4.3 from \cite{blanchard2008statistical}, which 
	is a generalization	of Theorem~4.2 of \cite{massart2000some}, 
	to our functional DWD classifier.
	\\
	
	First, we choose $\mathcal{R} = \{
	M^{-1}A^{-1}2^k\,|\, k\in\natural,k\le\lceil\log_2n\rceil
	\}$ to be our countable set of radii for the balls 
	$B(R)$ with $R\in\mathcal{R}$.  To apply Theorem~4.3,
	we require a sequence $z_R$ and choose 
	$z_R = \log( \log_2n+2 )$ similarly
	to \cite{blanchard2008statistical}.  We also require a 
	penalty function that satifies
	$$
	\mathrm{pen}(R) \ge 250K\frac{r^*}{C_R} + 
	\frac{(KC_R+28b_R)(z_R+\xi+\log2)}{3n}
	$$
	where $b_R = 1+RAM$, 
	$C_R = \frac{RAM}{\eta_1}+\frac{5}{\eta_0}$, 
	and $r_R^* = \frac{8C_R^2}{\sqrt{n}}\inf_{d\in\natural}\left[
	\frac{2d}{\sqrt{n}} + \frac{A\eta_1}{M}\sqrt{\sum_{j>d}\lmb_j} 
	\right]$ under setting (C1), and 
	$b_R = 1+RAM$, 
	$C_R = {RAM}+\frac{1}{\eta_0}$, 
	and $r_R^* = 625C_R^{*2}(x^*)^2/M^2$
	under setting (C2).
	To achieve that, we take 
	$\gamma(n)$ to be $\frac{1}{\sqrt{n}}\inf_{d\in\natural}\left[
	\frac{2d}{\sqrt{n}} + \frac{A\eta_1}{M}\sqrt{\sum_{j>d}\lmb_j} 
	\right]$ under setting (C1) and $(x^*)^2/M^2$ under setting (C2),
	and we take $
	\lambda_n = C\left( \gamma(n) + 
	\frac{\log(\delta^{-1}\log n)\vee 1}{n} \right)
	$
	for some universal constant $C$.
	Consequently, we can choose
	$
	\mathrm{pen}(R)  = \lmb_n\left(
	RAM + k_0
	\right)
	$
	for another suitable positive constant $k_0$.
	\\
	
	Therefore, given the estimator $\hat{\beta}$ and corresponding
	estimator $\hat{f}$, 
	we have 
	$
	P_n V_q(\hat{f}) + \lmb_n AMJ(\hat{\beta}) \le
	P_n V_q(0) + \lmb_n AMJ(0) = 1.
	$
	Therefore, the regularization term $\lmb_nAMJ(\hat{\beta})\le1$
	and consequently, $\lmb_n\ge n^{-1}$.  
	Thus, $J(\hat{\beta})\le nA^{-1}M^{-1}$.
	Recalling that 
	$B(R) = \{ \beta\in\mathcal{W}_2^m(I)\,:\, J({\beta})\le R\}$,
	we have that $\hat{\beta}\in B(\hat{R})$ for 
	$\hat{R} = J(\hat{\beta})\vee(A^{-1}M^{-1})$ and,
	updating $k_0$ as necessary, we have that
	\begin{align*}
	P_n V_q(\hat{f}) + \mathrm{pen}(\hat{R})
	&\le
	P_n V_q(\hat{f}) +
	\lmb_n(
	{AMJ(\hat{\beta})\vee1}) + k_0
	)\\
	&\le 
	P_n V_q(\hat{f}) +
	\lmb_n(
	J(\hat{\beta})AM + k_0)\\
	&\le 
	\inf_{\beta\in\mathcal{W}_2^m}\left\{
	P_n V_q({f}_\beta) +
	\lmb_nJ({\beta})AM
	+
	\lmb_n k_0
	\right\}\\
	&= 
	\inf_{R>0}
	\inf_{\beta\in B(R)}\left\{
	P_n V_q({f}_\beta) +
	\lmb_n( RAM
	+ k_0)
	\right\}\\
	&\le
	\inf_{R\in\mathcal{R}}
	\inf_{\beta\in B(R)}\left\{
	P_n V_q({f}_\beta) +
	\lmb_n( RAM
	+ k_0)
	\right\}
	\end{align*}
	where the third inequality results from $\hat{f}$ being the
	minimizer of the regularized estimation proceedure.
	Application of the ``model selection theorem''
	\citep{massart2000some,blanchard2008statistical}
	gives that 
	with probability $1-\delta$ that 
	\begin{align*}
	L(\hat{f},f^*) &\le
	2\inf_{R\in\mathcal{R}}\inf_{\beta\in B(R)}\left\{
	L(f_\beta,f^*) + 2\lmb_n ( RAM
	+ k_0 )
	\right\}\\
	&\le
	2\inf_{A^{-1}M^{-1}\le R\le nA^{-1}M^{-1}}\inf_{\beta\in B(R)}\left\{
	L(f_\beta,f^*) + 2\lmb_n ( 2^{\lceil \log_2RAM\rceil}
	+ k_0 )
	\right\}\\
	&\le
	2\inf_{R\le nA^{-1}M^{-1}}\inf_{\beta\in B(R)}\left\{
	L(f_\beta,f^*) + 2\lmb_n ( 2^{\lceil (\log_2RAM)_+\rceil}
	+ k_0 )
	\right\}\\
	&\le
	2\inf_{R\le nA^{-1}M^{-1}}\inf_{\beta\in B(R)}\left\{
	L(f_\beta,f^*) + 2\lmb_n ( RAM\vee 1)
	+ k_0 )
	\right\}\\
	&\le
	2\inf_{\beta\in \mathcal{W}_2^m}\left\{
	L(f_\beta,f^*) + 2\lmb_n AMJ(\beta)
	+ \lmb_nk_1 + k_0 )
	\right\}
	\end{align*}
\end{proof}

\end{document}